\documentclass{article}

\usepackage[preprint,nonatbib]{neurips_2022} 

\usepackage[utf8]{inputenc} 
\usepackage[T1]{fontenc}    

\usepackage{url}            
\usepackage{booktabs}       
\usepackage{amsfonts}       
\usepackage{nicefrac}       
\usepackage{microtype}      
\usepackage{xcolor}         

\usepackage[colorlinks=true,linkcolor=black,citecolor=black,urlcolor=black]{hyperref}       
\makeatletter
\def\Hy@Warning#1{} 
\makeatother

\usepackage{relsize}

\def\CC{C\nolinebreak[4]\hspace{-.05em}\raisebox{.4ex}{\relsize{-2}{\textbf{++}}}}

\usepackage{calc}
\usepackage{enumitem}
\usepackage{float}

\usepackage{amsmath}
\usepackage{amsthm}

\usepackage[capitalize]{cleveref}

\usepackage{bbm}
\usepackage{caption}
\usepackage{pgfplots}
\usepackage{tikz}
\usetikzlibrary{trees, shapes}
\usetikzlibrary{shapes,decorations,arrows,calc,arrows.meta,fit,positioning}
\usetikzlibrary{shadows,shadings,shapes.symbols}
\tikzset{
    double color fill/.code 2 args={
        \pgfdeclareverticalshading[%
        tikz@axis@top,tikz@axis@middle,tikz@axis@bottom%
        ]{diagonalfill}{100bp}{%
            color(0bp)=(tikz@axis@bottom);
            color(50bp)=(tikz@axis@bottom);
            color(50bp)=(tikz@axis@middle);
            color(50bp)=(tikz@axis@top);
            color(100bp)=(tikz@axis@top)
        }
        \tikzset{shade, left color=#1, right color=#2, shading=diagonalfill}
    }
}
\usepgfplotslibrary{fillbetween}
\usepackage{graphicx}
\usetikzlibrary{decorations.pathreplacing, matrix}

\usepackage[boxed, linesnumbered, noend, noline]{algorithm2e}
\SetKwInput{KwData}{Input}
\SetKwInput{KwResult}{Output}

\renewcommand\log{\ln}

\newcommand\dist{\mathrm{dist}}

\newcommand{\vA}{\vec A}

\renewcommand{\epsilon}{\eps}

\renewcommand{\vec}[1]{\boldsymbol{#1}}

\newcommand\SIGMA{\vec\sigma}

\newtheorem{definition}{Definition}[section]

\newtheorem{theorem}[definition]{Theorem}
\newtheorem{lemma}[definition]{Lemma}
\newtheorem{proposition}[definition]{Proposition}

\newtheorem{fact}[definition]{Fact}

\newcommand\cC{\mathcal{C}}
\newcommand\cD{\mathcal{D}}

\newcommand\cE{\mathcal{E}}

\newcommand\cT{\mathcal{T}}

\newcommand\cV{\mathcal{V}}

\def\cE{{\mathcal E}}

\newcommand\eps{\varepsilon}

\newcommand\Erw{\mathbb{E}}

\newcommand{\Po}{{\rm Po}}
\newcommand{\Bin}{{\rm Bin}}

\newcommand\bc[1]{\left({#1}\right)}
\newcommand\cbc[1]{\left\{{#1}\right\}}

\newcommand\brk[1]{\left\lbrack{#1}\right\rbrack}

\newcommand\abs[1]{\left|{#1}\right|}

\newcommand\pr{\mathbb{P}} 
\renewcommand\Pr{\pr}

\newcommand\A{\vA}

\def\pr{{\mathbb P}}

\usepackage{lipsum}

\newcommand{\remove}[1]{}

\newcommand{\be}{\begin{equation}}
    \newcommand{\bel}[1]{\begin{equation}\lab{#1}\ }
        \newcommand{\ee}{\end{equation}}
    \newcommand{\bea}{\begin{eqnarray}}
        \newcommand{\eea}{\end{eqnarray}}
    \newcommand{\bean}{\begin{eqnarray*}}
        \newcommand{\eean}{\end{eqnarray*}}

\graphicspath{{..}}
\pgfplotsset{compat=1.14}

\tikzstyle{node} = [shape=circle,draw=black]

\newcommand{\colOrange}{orange}
\newcommand{\colGreen}{green}
\newcommand{\colBlue}{blue}
\tikzstyle{node} = [shape=circle,draw=black]
\tikzstyle{dottetReplace} =[fill = violet!20!white]
\newcommand{\betweenDist}{0}

\newcommand{\pz}{\vec{\omega}}
\newcommand{\cc}{\vec{\omega}_c}
\newcommand{\activenodes}{\vec{X}^\star}
\newcommand{\candidates}{\mathcal{C}}
\newcommand{\distactive}[1]{t_{#1}^{\activenodes}}
\newcommand{\lca}{\cc}

\title{Inference of a Rumor's Source\\ in the Independent Cascade Model}

\author{%
  Petra Berenbrink \\
  Universität Hamburg\\
  Hamburg, Germany \\
  \texttt{petra.berenbrink@uni-hamburg.de} \\
  \And
  Max Hahn-Klimroth \\
  TU Dortmund University\\
  Dortmund, Germany \\
  \texttt{max.hahn-klimroth@cs.tu-dortmund.de} \\
  \And
  Dominik Kaaser \\
  TU Hamburg\\
  Hamburg, Germany\\
  \texttt{dominik.kaaser@tuhh.de} \\
  \And
  Lena Krieg \\
  TU Dortmund University\\
  Dortmund, Germany \\
  \texttt{lena.krieg@tu-dortmund.de} \\
  \And
  Malin Rau \\
  Universität Hamburg\\
  Hamburg, Germany \\
  \texttt{malin.rau@uni-hamburg.de} \\
}

\makeatletter
\def\@noticestring{}
\makeatother
\begin{document}

\maketitle
\pagestyle{plain}

\begin{abstract}
We consider the so-called \emph{Independent Cascade Model} for rumor spreading or epidemic processes popularized by Kempe et al.\ [2003].
In this model, a small subset of nodes from a network are the source of a rumor.
In discrete time steps, each informed node "infects" each of its uninformed neighbors with probability $p$.
While many facets of this process are studied in the literature, less is known about the inference problem: given a number of infected nodes in a network, can we learn the source of the rumor?
In the context of epidemiology this problem is often referred to as \emph{patient zero problem}.
It belongs to a broader class of problems where the goal is to infer parameters of the underlying spreading model, see, e.g., Lokhov [NeurIPS'16] or Mastakouri et al. [NeurIPS'20].

In this work we present a maximum likelihood estimator for the rumor's source, given a snapshot of the process in terms of a set of active nodes $X$ after $t$ steps. Our results show that, for cycle-free graphs,  the likelihood estimator undergoes a non-trivial phase transition as a function $t$. We provide a rigorous analysis for two prominent classes of acyclic network, namely $d$-regular trees and Galton-Watson trees, and verify empirically that our heuristics work well in various general networks. 
\end{abstract}

\section{Introduction}

In this paper we consider a stochastic diffusion process which models the spread of information or influence in networks. Influence propagation is  motivated by many applications from various fields: in marketing these processes are studied to maximize the adoption of a new product; these processes  are used to study how social media influencers manipulate humans in social networks, in epidemiology they are used to study the spread of viruses or disease \cite{Brauer2017,SIRFIRST1927}, or, in general as processes that spread information in networks \cite{Becker_Coro_DAngelo_Gilbert_2020, goffman1964generalization, Kempe2003,Lerman_Ghosh_2010,rum2007,Sadilek_Kautz_Silenzio_2021}.  

On a very high level these processes work as follows. Initially a small subset of the vertices are in a distinguished state (they might have  a piece of information, or they are infected, depending on the application in mind).
In this paper we will call these vertices, based on rumor spreading, \emph{informed} vertices. Informed vertices can inform their neighbors and the rumor spreads as time passes by through the network. 
Most publications studying stochastic diffusion processes observe these processes  in a \emph{forward direction}, i.e., they consider how information spreads in a network, how many nodes will become informed,  or which nodes have the largest influence on the vertices in a social network. Those processes are well understood in simple networks \cite{britton_janson_martin-loef_2007,sir_random_rigor}. In this paper we study the learning problem of inferring $\pz$, a problem which received far less attention.  Our goal is to detect the \emph{source} of the rumor. In the disease spreading settings this problem is referred to as  the \emph{patient zero problem}.
This inference problem was studied with respect to the SI model from epidemiology rigorously \cite{Shah2012}. Under the simple SI disease spreading model once infected nodes stay infected and they can infect randomly chosen neighboring nodes. Understanding inference problems of this kind better will help us to find the source of outbreaks of infectious diseases like COVID-19 or to find the source of rumors. The later might help to prevent that political elections are influenced from the outside world.

In this paper we employ the well-known \emph{Independent Cascade model} (see, e.g., \cite{Kempe2003}).
The process starts with an initial set of active nodes $I_0$ and it works in discrete steps.  When node $v$ first becomes active in step $t$, it is given a single chance (one shot) to activate each (currently inactive) neighbor with probability $p$. Whether or not $v$ succeeds to activate any of its neighbors, it cannot make any further attempts in later rounds. It remains inactive for all steps $t'>t$. The process runs until no more activations are possible. 
We assume that $I_0$ (the \emph{rumor's source}) contains only a single node denoted by $\pz$. 
We call all nodes that were activated at one point of time during the process \emph{informed}. 
Note that this  \emph{one-shot} property is  a very fitting model for rumor or disease spreading in social networks. Indeed, once a user hears about an article supporting her opinion, she will either ignore it or share it within her social contacts in the near future. Every of the possible recipients either ignores her opinion (does not get activated) or decides to share it again with its peers (gets activated). Furthermore, users are unlikely to share the same article twice. In the case of disease spreading the informed vertices model the persons which caught the illness and the active ones model the persons which are infectious at any point of time. 

Our problem fits very well into the so-called  \emph{teacher-student model}, a framework which is frequently used to model (machine) learning tasks. The model was introduced by Gardner \cite{Gardner_1989} in the context of studying fundamental properties of the Perceptron, a fairly simple binary classifier. \cite{coja_spiv,NEURIPS2021_4e8eaf89,NEURIPS2021_a2137a2a,lenka_florent_stat_physics_inference}. Suppose a \emph{teacher} samples a \emph{ground-truth} $\SIGMA$ from a distribution called \emph{teacher's prior} $\mu_P$. Rather than directly revealing this ground-truth to a student, the teacher creates a condensed version $\hat \SIGMA$ of the ground-truth by the means of a \emph{teacher's model} $\mu_M$. Now, in the so-called Bayes optimal case, the teacher conveys $\hat \SIGMA, \mu_P, \mu_M$ to her student. The student's task is to infer a non-trivial guess of $\SIGMA$ from the observed data. 
In the context of our contribution, the teacher's prior is, given a network $G$, to select one node uniformly at random as the rumor's source $\pz$.
Suggestion: The teacher's model is the $t$-step forward process of the Independent Cascade Model on $G$ starting with the source $\pz$.
The condensed information $\hat \SIGMA$ that the student receives consists of the network $G$ and the set of currently active nodes $\activenodes$. The student's task now is to infer $\pz$ from $G$ and $\activenodes$. 

In our first result, we prove that our model is Bayes optimal (or, in terms of statistical physics, the \emph{Nishimori property} holds). 
Secondly, with respect to $d-$regular trees, we show that for a small spreading parameter $p$ (in comparison to the degree of the tree $d$) it is not possible to infer the source of the rumor. 
For a large spreading parameter $p$ strong detection is possible, we show that we can detect the source node with a very large probability. 
The probability approaches one as a function of the time $t$ at which the snapshot was done. 
For intermediate $p$ we show that a proposed node $\cc$ is the source of the rumor with a constant probability.
Furthermore, we bound the probability that the real rumor source $\pz$ is far away from $\cc$. 
Another way to read our results is that inference of $\pz$ gets easier for increasing value of $(d-1) \cdot p$. 
Finally, we establish a similar phase transition with respect to Galton-Watson processes with spreading parameter $\Po( \lambda )$.

\paragraph{Related Work.}
Forward propagation processes, like the epidemic models \cite{Brauer2017,SEIR,SIRFIRST1927,SIS}, rumor spreading \cite{goffman1964generalization,Lerman_Ghosh_2010,rum2007,Sadilek_Kautz_Silenzio_2021}, information cascades \cite{gruhl2004,Kempe2003,Zhao2011}, blog propagation models \cite{Leskovec2007PatternsOC}, and marketing strategies \cite{Becker_Coro_DAngelo_Gilbert_2020, Kempe2003} have been studied extensively and for a long time within different research communities and we refrain from discussing the extensive literature here. On the contrary, rigorous contributions on the corresponding inference problem, the source detection task, are very rare. 

To the best of our knowledge, all existing rigorous results on how to infer a rumor's source, are with respect to the SI model \cite{kazemitabar2019approximate, Shah2010, Shah2012}. In our notation, the author's of \cite{Shah2010, Shah2012} prove that on specific infinite acyclic networks like $d$-regular trees, super-critical Galton-Watson processes, and geometric graphs, approximate inference of $\pz$ is always possible in the SI model as long as the infinite tree satisfies certain expansion properties (for instance, in the $d$-regular case, this requires $d \geq 3$). Furthermore, there are strong tail-bounds given that prove that the probability of declaring a \emph{far-apart} node as the rumor's source is small. The results are proved by a fundamental connection between a generalized Polya's urn and the SI model on acyclic networks.

Nevertheless, there are many well explained heuristics towards the source detection task on various network types which are supported by extensive simulation studies \cite{Bindi2017,Jain2016,Ji2017AnAF,Prakash2013,Wang_Wang_Pei_Ye_2017}, in recent contributions also based on neural networks \cite{biazzo2021epidemic,Shah2020FindingPZ,Shu2021}. Finally, a closely related problem that attained attention recently, is not to infer $\pz$ given $\activenodes$ but to infer the parameters of the underlying spreading model. Over the last years, learning strategies towards this problem were proposed and studied experimentally \cite{reconstruction_epidemics_parameters,spread_covid_neurips}.

\section{Model and Results}
We are given a communication network $G = (V,E)$ with $n$ vertices $x_1 \ldots x_n$ and $m$ edges.  We assume $I_0=\{\pz\}$ and $\pz$ is chosen uniformly at random from all vertices. Hence, the rumor originates at a single source.
A \emph{spreading parameter} $p \in (0,1)$ determines the  viciousness of the rumor.

The diffusion process in the Independent Cascade Model runs in discrete, synchronous steps.
Let $I_t$ be the set of vertices which are active in step $t$.
In step $t$ we will call the vertices in $I_0\cup I_1\ldots \cup I_t$ \emph{informed}.
In every step $t\ge 1$ every  active node $x_i \in I_t$ activates any of its uninformed neighbors with probability $p$.  
All these  newly informed vertices form  the set $I_{t+1}$.  Note that  every node becomes active exactly once but vertices have potentially the chance to be activated by each of their neighbors.  
Note that it can happen that the process dies out at a step $t$. 
In this case for all $t' \geq t$ it holds that $I_{t'}= \emptyset$.

The \emph{interference problem} is defined as follows.
We observe the state of the network at an arbitrary time $t$. 
The task is to infer $\pz$ given only $(G, I_t)$ and the parameter $p$.
In this paper  we study the following variants of the problem.
\begin{itemize}
    \item \emph{Strong detection:}   here we have to infer $\pz$ correctly with high probability\footnote{We say that a sequence of events $\cE_1, \cE_2, \ldots$ takes place with high probability (w.h.p.) if $\lim_{t \to \infty} \Pr\bc{\cE_t} = 1$.},
    \item\emph{Weak detection:}  inference of $\pz$ is only required with positive probability\footnote{We say that  a sequence of events $\cE_1, \cE_2, \ldots$ takes place with positive probability if $\lim_{t \to \infty} \Pr\bc{\cE_t} > 0$.}.
\end{itemize}

Our first result relates the probability that a given set $X$ is active conditioned on some node $v$ being the source to the probability that $v$ is the source conditioned on $X$ being active. 
This establishes the so-called \emph{Nishimory property} (or Bayes optimality) of our inference problem. In terms of the teacher-student model, it states that the student has access to the teacher's prior and the teacher's model. Equivalently, in expectation there is no statistical difference between the ground-truth and a uniform sample from the posterior distribution \cite{lenka_florent_stat_physics_inference}. \Cref{thm_mle_general} applies to all types of networks as long as the rumor's source $\pz$ is chosen uniformly at random.

\begin{theorem} \label{thm_mle_general} Let $G = (V, E)$ be an arbitrary network and  fix an arbitrary step $t$. 
Let $\activenodes$ be the set  active vertices  in step $t$. For any $X\subseteq V$
\begin{align*}
\arg\max_{v \in V} \Pr \bc{ \activenodes = X \mid \pz = v} = \arg\max_{v \in V} \Pr \bc{ \pz = v \mid \activenodes = X}.
\end{align*}
\end{theorem}

The above theorem is used to show the main results of this paper.
It allows us to calculate  $\Pr \bc{ \activenodes= X \mid \pz = v}$ instead of $\Pr \bc{ \pz = v \mid \activenodes=X }$,  which often is  more accessible.
Note that calculating the first probability is quite challenging in general networks $G$ since the entropy of $\activenodes$ is very large.

For the remaining analytical part of the of the paper we consider acyclic networks, namely $d$-regular trees and Galton-Watson trees with offspring distribution $\Po(\lambda)$.
Galton-Watson trees with offspring distribution $\cD$ are defined by the following experiment. 
We start with one node which spawns $\vec d_0 \sim \cD$ children. 
Recursively, any of the children $w_1, \ldots w_{\vec d_0}$ spawns $\vec \eta_1, \ldots, \vec \eta_{\vec d_0} \sim  \cD $ children (and so on). 
We call such a process super-critical, if with positive probability, the vertices spawned during the process form an infinite tree.
It is well known that a (not too dense) instance of an Erdős–Rényi random graph $\mathbb{G} \bc{ n, \frac{d}{n} }$ \emph{locally} looks like a Galton-Watson tree with offspring distribution $\Po(d-1)$.
In contrast, a random $d$-regular graph looks locally like a $d$-regular tree, provided $d$ is not too large~\cite{Frieze2016}.

We assume that a teacher fixes a time $t$ at which she observes the network. We assume that the student is not aware of the time when the process started. We define  $\activenodes$ as the set of nodes which are active after $t$ steps of the Independent Cascade model (starting from an unknown and randomly chosen source $\pz$).  Note that $\activenodes$ can be empty. 
The set of candidate nodes $\candidates$ are all nodes that have the same distance in $G$ to each node in $\activenodes$.
Note that the set $\candidates$ is not empty since $\pz\in \candidates$. 
The \emph{closest candidate} $\cc \in \candidates$ is defined as the node with minimum distance to all nodes in $\activenodes$. 
Our \emph{source detection heuristic} calculates the closest candidate $\cc\in \candidates$ and returns it as the estimated rumors source.
If  $\activenodes=\emptyset$  (the process died out before time $t$) or contains at most one node, the heuristic returns a failure.
The following results describe the probability of success or failure for this source detection heuristic with respect to the models parameter.

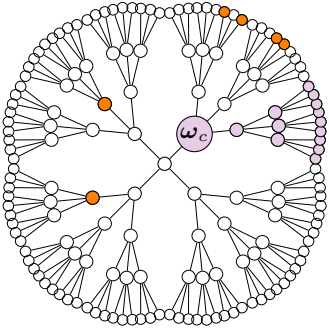
\begin{figure}[t]
    \centering
    \begin{minipage}{0.4\textwidth}



\tikzstyle{level 1}=[sibling angle=90,node, minimum size=0.5em,inner sep =0]
\tikzstyle{level 2}=[sibling angle=39.5, node, minimum size=0.5em,inner sep =0]
\tikzstyle{level 3}=[sibling angle=19, node, minimum size=0.5em,inner sep =0]
\tikzstyle{level 4}=[sibling angle=13.2, node, minimum size=4pt,inner sep =0]
\tikzstyle{edge from parent}=[draw, very thin]

\newcommand{\infNodes}{ k32,  k41}
\newcommand{\infLeafs}{k2222,k2223,k2311,k2313}
\newcommand{\originNode}{k2}
\newcommand{\candSub}{1}

\begin{tikzpicture}[decoration=brace,
  grow cyclic,
  shape = circle,
  level distance=2em,
  scale = 0.8]

\node[circle, minimum size=1em, level 1] at (-\betweenDist,0) {} child [level 1] foreach \A in {k4, k1, k2, k3}
    { node[level 1](\A) {} child [level 2] foreach \B in {\A1, \A2, \A3}
        { node[level 2](\B) {} child [level 3] foreach \C in {\B1, \B2, \B3}
            { node[level 3](\C) {} child [level 4] foreach \D in {\C1, \C2, \C3}
                {node[level 4](\D) {}} 
            }
        }
    };

\node[level 1,dottetReplace, inner sep = 0.2mm] (origin) at (\originNode) {{$\footnotesize\lca$}};

\node[level 2,dottetReplace] (origin) at (\originNode\candSub) {};
\foreach \x in {1,2,3}{
    \node[level 3, dottetReplace] () at (\originNode\candSub\x){};
    \foreach \y in {1,2,3}{
        \node[level 4, dottetReplace] () at (\originNode\candSub\x\y){};
    }
}

\foreach \x in \infLeafs
{
        \node[level 4, fill=\colOrange] (test) at (\x) {};
}
\foreach \x in \infNodes
{
        \node[level 2, fill=\colOrange] (test) at (\x) {};
}


\end{tikzpicture}
    \end{minipage}\begin{minipage}{0.6\textwidth}
    \caption{
    Visualization of a possible snapshot of the spreading process. 
    Here, $\cc$ spawned four sub-trees out of which three contain active elements of $\activenodes$ (orange nodes) and one does not contain active elements (purple). 
    Thus, the candidate set $\candidates$ of possible rumor's sources consists of all vertices in the purple sub-tree. 
    Note that here only a finite part of the infinitely expanding 4-regular tree is presented.
    }
    \label{fig:circltree}
    \end{minipage}
\end{figure}

The first theorem shows a phase transition between  weak and strong detection for $d$-regular trees.
We show that, for small $p\cdot(d-1) $ it is not possible to infer the source of the rumor. 
This is due to the huge likelihood that, in this setting,  the process dies out and $\activenodes=\emptyset$ or the corresponding set $\activenodes$ is very small which makes the inference impossible. 

For large $p\cdot (d-1)$ we show that the source node is the closest candidate with probability $1 - o_d(1)$.
Note that in this scenario, each active node will infect several nodes and it is unlikely that the process dies out.
The size of $\activenodes$ grows as a function of $d$ and $t$ which makes the prediction more and more reliable. 
For intermediate $p\cdot(d-1)$ we show that the closest candidate $\cc$ is the source of the rumor with a constant probability. 
Furthermore, we bound the probability that the real rumor source $\pz$ is far away from $\cc$. 
Intuitively this means that inference of $\pz$ gets easier for increasing values of $(d-1) \cdot p$.

\begin{theorem}[$d$-regular trees]\label{thm_lca_dreg}
Let $G = (V, E)$ be an infinite $d$-regular tree and let  $\activenodes$ be the set of active nodes generated 
by the Independent Cascade Process with spreading parameter $p$ after $t = \omega(1)$ steps. 
Then, the following phase-transitions occur.
\begin{itemize}[nosep]
    \item If $(d-1)\cdot p \leq 1$, any estimator fails at weak detection with probability $1-o_t(1)$.\footnote{We denote by $o_t(1)$ a quantity that tends to zero with $t \to \infty$.}
    \item If $1 < (d-1)\cdot p = \Theta(1)$ then the closest candidate $\cc$ is the source  of the rumor $\pz$ with constant probability (weak detection). 
    Furthermore, the probability that $\dist(\cc, \pz) > k$ is at most $\exp \bc{ - \Omega(k)}$.
    \item If $(d-1)\cdot p = \omega(1)$ then closest candidate $\cc$ is the source of the rumor $\pz$  with probability $1 - o_{d}(1)$ (strong detection).
\end{itemize}
\end{theorem}
Since, in expectation, a $\Po(\lambda)$-Galton-Watson tree with $\lambda=d-1$ spawns exactly as many children as a $d$-regular tree, it might be not too surprising that a similar phase transition occurs in this model.  

\begin{theorem}[Galton-Watson processes]\label{thm_lca_gw}
Let $G = (V, E)$ be an infinite tree generated by a $\Po(\lambda)$-Galton-Watson process. Let  $\activenodes$ be the set of active nodes generated 
by the Independent Cascade Process with spreading parameter $p$ after $t = \omega(1)$ steps.  Then, the following phase-transition occurs.
\begin{itemize}[nosep]
    \item If $\lambda p \leq 1$, any estimator fails at weak detection with probability $1-o_t(1)$.
    \item If $1 < \lambda p = \Theta(1)$, then the closest candidate $\cc$ is the source  of the rumor $\pz$ with positive probability (weak detection).
    Furthermore, the probability that $\dist(\lca, \pz) > k$ is at most $\exp \bc{ - \Omega(k)}$.
    \item If $\lambda p = \omega(1)$, then closest candidate $\cc$ is the source of the rumor $\pz$  with probability $1 - o_{\lambda}(1)$  (strong detection).
\end{itemize}
\end{theorem}
In \cref{thm_lca_dreg,thm_lca_gw} we assume that the underlying tree network is infinitely large. This is conceptually necessary. Indeed, the trivial algorithm that outputs one node uniformly at random succeeds at weak detection in finite networks. In the next section, we provide extensive simulations that verify the asymptotic statements of the theorems on small networks. Furthermore, we present simulation results on non-acyclic networks such as random geometric graphs and show that (the natural extension of) the closest candidate heuristics works well.

The main proof strategy of \cref{thm_lca_dreg,thm_lca_gw} is to interpret the transmission process as a special type of percolation on the underlying network. As in \cref{thm_lca_gw} the underlying network itself is random, it turns out to be technically non-trivial to pin down the exact distributions involved in this process due to subtle rare events that might yield either very small or large node degrees. Fortunately, the Poisson thinning technique allows us to carry out the calculations smoothly.

\section{Simulations} \label{sec:simulations}

In this section we present simulation results that support and complement our main theorems for cyclic graphs. 
Our simulation software is implemented in the \CC{} programming language.
As a source of randomness we use the Mersenne Twister \texttt{mt19937\_64} provided by the \CC{11} \texttt{\textless random\textgreater{}} library.
All simulations have been carried out on four machines equipped with two Intel(R) Xeon(R) E5-2630 v4 CPUs and 128 GiB of memory each, running the Linux 5.13 kernel.

A simulation run consists of three parts. First, we generate a network $G = (V, E)$.
To this end, we have implemented generators for Erdős-Rényi graphs, random $d$-regular graphs (configuration model \cite{Janson_2011}) and random geometric graphs \cite{Penrose2003}.
Secondly, we run the Independent Cascade Process for $t$ rounds  starting from a randomly chosen node $\pz \in V$.
Finally, we generalize our source detection algorithm to cyclic graphs as follows.
For $v\in V$ and $t'\ge 0$ let $N_{t'}(v)$ denote the set of nodes $w \in V$ that have distance at most $t'$ to $v$.
Then we calculate 
\[ N_{t'}=\bigcap_{u \in \activenodes} N_{t'}(u).\]
Hence, $N_t'$ is the set of nodes with distance at most $t'$ to 
every node in $\activenodes$.
We pick the minimum $t'$ such that $N_{t'} \neq \emptyset$ and return $N_{t'}$ as the candidate set.
Note that such a $t'$ exists since $\pz \in N_t$.

To generate our data we execute 100 independent simulation runs for each network where we simulate the Independent Cascade Model for $8$ rounds.
Our data for Erdős-Rényi graphs with $n = 10^5$ nodes and expected node degree $4$ are shown in \cref{tab:table1}.
Similar data for random $4$-regular graphs in the configuration model and $2$-dimensional random geometric graphs with expected node degree $16$ can be found in \cref{apx:additional-data}. 
The tables show for various spreading probabilities $p$ the number of successes in detecting the source, the numbers of errors for the two error cases $\pz \not\in N_{t'}$ and $\activenodes = 0$, and the average and maximum distance of $\pz$ to nodes in $N_{t'}$.
The success rates are also visualized in \cref{fig:plot-1}.
We remark that our empirical data confirm the phase transitions as claimed in \cref{thm_lca_dreg,thm_lca_gw} for the even more general case of locally tree-like graphs, see below.

\begin{figure}[p]
\begin{table}[H]
\caption{Simulation results for Erdős-Rényi graphs.}
\centering
\label{tab:table1}
 \small
\begin{tabular}{lrrrrr}
\toprule
$p$ & \parbox{\widthof{number of}}{\raggedleft number of\newline successes}& $\cc \neq \pz$ & $\activenodes = \emptyset$ & \parbox{\widthof{distance}}{\raggedleft average\newline distance} & \parbox{\widthof{maximum}}{\raggedleft maximum\newline distance}\\
\midrule
0.00 & 0 & 0 & 100 & - & -\\
0.05 & 0 & 0 & 100 & - & -\\
0.10 & 0 & 0 & 100 & - & -\\
0.15 & 0 & 1 & 99 & 6.00 & 6\\
0.20 & 0 & 6 & 94 & 7.50 & 9\\
0.25 & 0 & 14 & 86 & 6.63 & 8\\
0.30 & 2 & 30 & 68 & 7.34 & 10\\
0.35 & 11 & 35 & 54 & 7.23 & 10\\
0.40 & 21 & 49 & 30 & 5.87 & 9\\
0.45 & 33 & 43 & 24 & 6.14 & 9\\
0.50 & 42 & 33 & 25 & 1.15 & 8\\
0.55 & 54 & 31 & 15 & 0.54 & 3\\
0.60 & 63 & 24 & 13 & 0.36 & 3\\
0.65 & 74 & 19 & 7 & 0.30 & 2\\
0.70 & 78 & 17 & 5 & 0.21 & 2\\
0.75 & 78 & 15 & 7 & 0.17 & 2\\
0.80 & 82 & 12 & 6 & 0.18 & 3\\
0.85 & 81 & 15 & 4 & 0.17 & 2\\
0.90 & 86 & 13 & 1 & 0.17 & 2\\
0.95 & 85 & 10 & 5 & 0.10 & 1\\
1.00 & 92 & 5 & 3 & 0.05 & 1\\
\bottomrule
\end{tabular}
\end{table}

\begin{figure}[H]
\def\ps{\small}
\begin{minipage}{1.8in}\input{plot1}\end{minipage}\begin{minipage}{\textwidth-1.8in}
\caption{Visualization of success rates for Erdős-Rényi graphs, random regular graphs, and random geometric graphs.}
\label{fig:plot-1}
\end{minipage}
\end{figure}

\begin{figure}[H]
\def\ps{\small}
\begin{minipage}{1.8in}\input{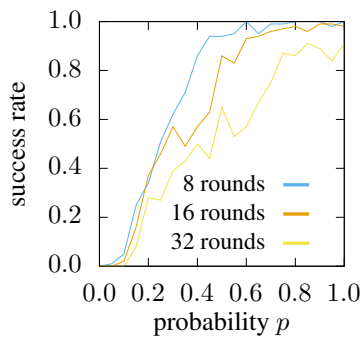}\end{minipage}\begin{minipage}{\textwidth-1.8in}
\caption{Visualization of success rates after $8$, $16$, and $32$ rounds of the Independent Cascade Model.}
\label{fig:plot-2}
\end{minipage}
\end{figure}

\begin{figure}[H]
\def\ps{\small}
\begin{minipage}{1.8in}\input{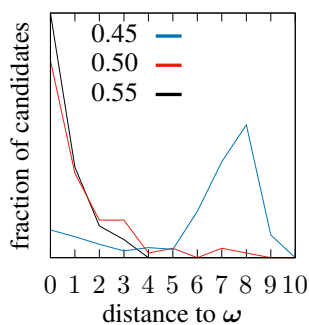}\end{minipage}\begin{minipage}{\textwidth-1.8in}
\caption{Histogram of the distribution of the distances of the candidates returned by our heuristic to $\pz$ for $p = 0.45, 0.5, 0.55$.}
\label{fig:plot3}
\end{minipage}
\end{figure}

\end{figure}

\paragraph{Varying Numbers of Rounds.}
In \cref{fig:plot-2} present simulation data for random geometric graphs with expected node degree $16$, where we increased the numbers of rounds of the Independent Cascade Model. We show the success rates after $8$, $16$, and $32$ rounds.

\paragraph{Phase Transitions.}
Finally in \cref{fig:plot3} we present an additional plot that highlights the phase transition behavior of our algorithm.
Recall that our heuristic in cyclic networks may return more than one candidate. 
In our simulation we compute the distance for each candidate found by our heuristics to $\pz$. 
The plot shows a histogram of these distances for three values of $p$, namely $0.45, 0.50,$ and $0.55$. 
From this histogram, a phase transition at $p = 0.5$ becomes eminent: while with $p = 0.45$ the majority of the candidates as a large distance to $\pz$, this changes drastically for $p = 0.55$, where the largest distance found is as small as $3$.

\section{Proofs}
In this section, we prove \cref{thm_mle_general,thm_lca_dreg}. Furthermore, the proof of \cref{thm_lca_gw} is sketched as it is similar to the one of \cref{thm_lca_dreg}. All omitted details can be found in the appendix.

\subsection{Proof of \cref{thm_mle_general}}
Observe that by definition, we have for any $v, w \in V$ that $ \Pr \bc{ \pz = v } = \Pr \bc{ \pz = w }$.
Thus, Bayes' rule and the law of total probability implies
\begin{align*}
    \Pr \bc{ \pz = v \mid \activenodes = X} &= \frac{ \Pr \bc{ \activenodes = X \mid \pz = v} \Pr \bc{ \pz = v }}{ \Pr \bc{ \activenodes = X } }  \\
    &= \frac{ \Pr \bc{ \activenodes = X \mid \pz = v} \Pr \bc{ \pz = v }}{ \sum\limits_{\omega \in V} \Pr \bc{ \activenodes = X \mid \pz = \omega} \Pr \bc{ \pz = w }} = \frac{ \Pr \bc{ \activenodes = X \mid \pz = v} }{ \sum\limits_{\omega \in V} \Pr \bc{ \activenodes = X \mid \pz = \omega}}.
\end{align*}
As $\sum_{\omega \in V} \Pr \bc{ \activenodes = X \mid \pz = \omega}$ is independent from $v$, we have
\begin{align*}
    \arg\max_{v \in V}~ \Pr \bc{ \pz = v \mid \activenodes = X} &= \arg\max_{v \in V}~ \Pr \bc{ \activenodes = X \mid \pz = v}
\end{align*}
and the theorem follows. \qed

\subsection{Proof of \cref{thm_lca_dreg}}
A crucial observation in the proof of \cref{thm_lca_dreg} is the following. 
If node $v$ gets activated during the spreading process by node $w$, it has $d-1$ additional neighbors $v_1, \ldots, v_{d-1}$ except $w$ which we call \emph{children} of $v$. 
Any of those children gets activated with probability $p$ independently from everything else in the next step. 
Suppose without loss of generality that $v_1, \ldots, v_{\vec d_0}$ are the activated children where $\vec d_0 \sim \Bin( d-1, p )$. 
In every of those children $v_i$, a new and independent rumor spreading process starts in the tree rooted at $v_i$ and directed away from $v$. 
As this tree is, itself, $d$-regular, this process is distributed equally as starting $\vec d_0$ independent Galton-Watson processes with offspring distribution $\Bin \bc{ d-1, p }$. 
Depending on $(d-1) p$, few, some or many of those processes will die out eventually.

To prove our result, we need some additional notation, see \cref{fig:Notation_Tree}. 
Given a node $v$, we can direct the tree away from $v$ and denote the set of subtrees rooted at $v$'s children by $\cT^v$. 
Most interesting to our proof are the subtrees that contain active nodes. 
We denote them by $\cT_{\activenodes }^{v}$ and denote 
\[Y_v = Y_v( \activenodes ) = \abs{\cT_{\activenodes }^{v}}.\] 
Note that all candidates but one have at most one subtree containing active nodes. 
Only the closest candidate can have more than one.
Finally, let $ \distactive{v}$ denote the distance from $v$ to any of the vertices in $\activenodes$.

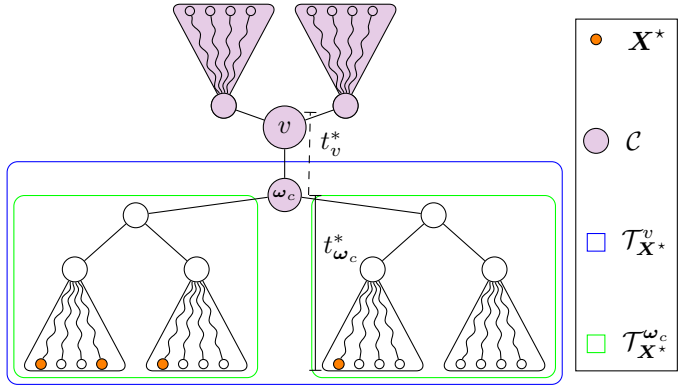
\begin{figure}[ht]
    \centering
    \begin{minipage}{0.7 \textwidth}
    \tikzstyle{smallNode} = [shape=circle,minimum size = 4pt, draw=black, inner sep =0, fill = \colOrange]
\tikzstyle{smallNodeNotInf} = [shape=circle,minimum size = 3.5pt, draw=black, inner sep =0]
\newcommand{\tri}[4]{ \draw[rounded corners, draw=black,#4] #1--#2--#3 --#1; } 
\newcommand{\vDiff}{0.8}
\newcommand{\hDiff}{0.9}
\newcommand{\triWidth}{0.8}
\newcommand{\triHeight}{1.5}
\newcommand{\padding}{0.1}
\newcommand{\smallwidth}{0.3}


\tikzset{snake it/.style={decorate, decoration={snake,amplitude=.3mm, segment length=3mm}}}

\begin{tikzpicture}[decoration=brace,
  grow cyclic,
  shape = circle,
  level distance=2em,
  scale = 0.9]

    \node[node,dottetReplace, inner sep =1] (root) at (0,0) {\tiny$\lca$};
    
    \node[node] (c1) at (-2.2,-0.3) {};
    \node[node] (c11) at ($(c1)+(-\hDiff,-\vDiff)$) {};
    \node[node] (c13) at ($(c1)+(\hDiff,-\vDiff)$) {};
    \draw(c1) edge (c11);
    \draw(c1) edge (c13);
    \tri{(c11)}{+(-\triWidth,-\triHeight)}{+(\triWidth,-\triHeight)}{}
    \tri{(c13)}{+(-\triWidth,-\triHeight)}{+(\triWidth,-\triHeight)}{}
     \draw[\colGreen, rounded corners] ($(c1.north)+(-\hDiff,0)+(-\triWidth,0)+(-\padding,\padding)$) rectangle ($(c1)+(\hDiff,-\vDiff)+(\triWidth,-\triHeight)+(\padding,-\padding)$);
     
    \node[smallNode] (s12) at ($(c11) + (-\triWidth,-\triHeight) + (\smallwidth,0.1) $) {};
    \node[smallNodeNotInf] (s13) at ($(c11) + (-\triWidth,-\triHeight) + (2*\smallwidth,0.1) $) {};
    \node[smallNodeNotInf] (s14) at ($(c11) + (-\triWidth,-\triHeight) + (3*\smallwidth,0.1) $) {};
    \node[smallNode] (s15) at ($(c11) + (-\triWidth,-\triHeight) + (4*\smallwidth,0.1) $) {};
    \draw[snake it] (c11) -- (s12);
    \draw[snake it] (c11) -- (s13);
    \draw[snake it] (c11) -- (s14);
    \draw[snake it] (c11) -- (s15);
    
    \node[smallNode] (s22) at ($(c13) + (-\triWidth,-\triHeight) + (\smallwidth,0.1) $) {};
    \node[smallNodeNotInf] (s23) at ($(c13) + (-\triWidth,-\triHeight) + (2*\smallwidth,0.1) $) {};
    \node[smallNodeNotInf] (s24) at ($(c13) + (-\triWidth,-\triHeight) + (3*\smallwidth,0.1) $) {};
    \node[smallNodeNotInf] (s25) at ($(c13) + (-\triWidth,-\triHeight) + (4*\smallwidth,0.1) $) {};
    \draw[snake it] (c13) -- (s22);
    \draw[snake it] (c13) -- (s25);
    \draw[snake it] (c13) -- (s23);
    \draw[snake it] (c13) -- (s24);
    
    \node[node] (c2) at (2.2,-0.3) {};
    \node[node] (c21) at ($(c2)+(-\hDiff,-\vDiff)$) {};
    \node[node] (c23) at ($(c2)+(\hDiff,-\vDiff)$) {};
    \draw(c2) edge (c21);
    \draw(c2) edge (c23);
    \tri{(c21)}{+(-\triWidth,-\triHeight)}{+(\triWidth,-\triHeight)}{}
    \tri{(c23)}{+(-\triWidth,-\triHeight)}{+(\triWidth,-\triHeight)}{}{}
     \draw[\colGreen, rounded corners] ($(c2.north)+(-\hDiff,0)+(-\triWidth,0)+(-\padding,\padding)$) rectangle ($(c2)+(\hDiff,-\vDiff)+(\triWidth,-\triHeight)+(\padding,-\padding)$);
     \draw[\colBlue, rounded corners] ($(c1.north)+(-\hDiff,0)+(-\triWidth,0)+(-\padding,0)+(-\padding,0)+(0,0.6)$) rectangle ($(c2)+(\hDiff,-\vDiff)+(\triWidth,-\triHeight)+(2*\padding,-2*\padding)$);

    \node[smallNode] (s32) at ($(c21) + (-\triWidth,-\triHeight) + (\smallwidth,0.1) $) {};
    \node[smallNodeNotInf] (s33)at ($(c21) + (-\triWidth,-\triHeight) + (2*\smallwidth,0.1) $) {};
    \node[smallNodeNotInf] (s34)at ($(c21) + (-\triWidth,-\triHeight) + (3*\smallwidth,0.1) $) {};
    \node[smallNodeNotInf] (s35)at ($(c21) + (-\triWidth,-\triHeight) + (4*\smallwidth,0.1) $) {};
    \draw[snake it] (c21) -- (s32);
    \draw[snake it] (c21) -- (s33);
    \draw[snake it] (c21) -- (s34);
    \draw[snake it] (c21) -- (s35);

    \node[smallNodeNotInf] (s42) at ($(c23) + (-\triWidth,-\triHeight) + (\smallwidth,0.1) $) {};
    \node[smallNodeNotInf] (s43)at ($(c23) + (-\triWidth,-\triHeight) + (2*\smallwidth,0.1) $) {};
    \node[smallNodeNotInf] (s44)at ($(c23) + (-\triWidth,-\triHeight) + (3*\smallwidth,0.1) $) {};
    \node[smallNodeNotInf] (s45)at ($(c23) + (-\triWidth,-\triHeight) + (4*\smallwidth,0.1) $) {};
    \draw[snake it] (c23) -- (s42);
    \draw[snake it] (c23) -- (s43);
    \draw[snake it] (c23) -- (s44);
    \draw[snake it] (c23) -- (s45);

    \node[node,dottetReplace] (c3) at ($(root)+(0,\vDiff)+(0,0.2)$) {$\scriptsize v$};
    \node[node,dottetReplace] (c31) at ($(c3)+(-\hDiff,0.4*\vDiff)$) {};
    \node[node,dottetReplace] (c33) at ($(c3)+(\hDiff,0.4*\vDiff)$) {};
    \draw(c3) edge (c31);
    \draw(c3) edge (c33);
    \tri{(c31)}{+(-\triWidth,\triHeight)}{+(\triWidth,\triHeight)}{dottetReplace}
    \tri{(c33)}{+(-\triWidth,\triHeight)}{+(\triWidth,\triHeight)}{dottetReplace}
    \node[node,dottetReplace] (c31) at (c31) {};
    \node[node,dottetReplace] (c33) at (c33) {};

    \node[smallNodeNotInf] (s52) at ($(c33) + (-\triWidth,\triHeight) + (\smallwidth,-0.1) $) {};
    \node[smallNodeNotInf] (s53)at ($(c33) + (-\triWidth,\triHeight) + (2*\smallwidth,-0.1) $) {};
    \node[smallNodeNotInf] (s54)at ($(c33) + (-\triWidth,\triHeight) + (3*\smallwidth,-0.1) $) {};
    \node[smallNodeNotInf] (s55)at ($(c33) + (-\triWidth,\triHeight) + (4*\smallwidth,-0.1) $) {};
    \draw[snake it] (c33) -- (s52);
    \draw[snake it] (c33) -- (s53);
    \draw[snake it] (c33) -- (s54);
    \draw[snake it] (c33) -- (s55);
    \node[smallNodeNotInf] (s62) at ($(c31) + (-\triWidth,\triHeight) + (\smallwidth,-0.1) $) {};
    \node[smallNodeNotInf] (s63)at ($(c31) + (-\triWidth,\triHeight) + (2*\smallwidth,-0.1) $) {};
    \node[smallNodeNotInf] (s64)at ($(c31) + (-\triWidth,\triHeight) + (3*\smallwidth,-0.1) $) {};
    \node[smallNodeNotInf] (s65)at ($(c31) + (-\triWidth,\triHeight) + (4*\smallwidth,-0.1) $) {};
    \draw[snake it] (c31) -- (s62);
    \draw[snake it] (c31) -- (s63);
    \draw[snake it] (c31) -- (s64);
    \draw[snake it] (c31) -- (s65);
    
     \draw[-|, dashed, rounded corners] ($(root.east)+(\padding, 0)$) -- node[above right =-0.9ex and 0.1ex]{$t^*_v$} ($(c3.north east)+(\padding, 0)+(0.05,0)  $);
     \draw[|-|] ($(root.east)+(0,0)+(0.2, 0 )$) -- node[above right=0.8ex and 0.1ex] {$t^*_{\lca}$} ($(root.east)+ (0,-0.3)+(0,0-\triHeight-\vDiff)+(0.2, 0 )$) ;
    

    \draw (root) edge (c1);
    \draw (root) edge (c2);
    \draw (root) edge (c3);

      
    \node[] (legend) at ($(c3) + (4*\triWidth, \triHeight) + (\hDiff, 0.5)+(0.4,-0.5)$) {};
    \draw[] ($(legend) + (-2*\padding, \padding)$) rectangle ($(legend) + 3*(\padding, -\padding)+ 4*(0.25, -1.2)$) ; 
    \node[smallNode, label={[label distance=4]0:$\activenodes$}] (activeNodes) at ($(legend) + (\padding, -2*\padding)$) {} ;
    \node[node, dottetReplace, label={[label distance=0.1]0:$\candidates$}] (candidates) at ($(legend) + (\padding, -2*\padding)+ (0, -1.5)$) {};
    \node[shape = rectangle, draw = \colBlue, label={[label distance=0.1]0:$\cT^{v}_{\activenodes}$}] (blue) at ($(legend) + (\padding, -2*\padding)+ 2*(0, -1.5)$) {};
    \node[shape = rectangle, draw = \colGreen, label={[label distance=0.1]0:$\cT_{\activenodes}^{\lca}$}] (blue) at ($(legend) + (\padding, -2*\padding)+ 3*(0, -1.5)$) {};
\end{tikzpicture}

    \end{minipage}\begin{minipage}{0.3 \textwidth}
    \caption{
    Here, $\cc$ spawned three sub-trees out of which two contain active elements of $\activenodes$ (orange) and one does not contain active elements (purple). Thus, $\candidates$ consists of all vertices in the purple sub-tree rooted at $\cc$.
    }
    \end{minipage}
    \label{fig:Notation_Tree}
\end{figure}

\begin{proposition}\label{prop_dreg}
If $x_t$ is the probability that at time $t$ there are no more active vertices under the spreading model on an infinite $d$-regular tree with infection probability $p$, we find
\[x_t = (1-p+p x_{t-1})^{d-1}\]
and the ultimate extinction probability of the spreading process is the smallest fixed-point of $x \mapsto (1 - p + p x)^{d-1}$.
Furthermore
{\small
\begin{align*}
    \Pr \bc{ Y_v \bc{ \activenodes } = k \mid \pz = v} = 
    \begin{cases}
        \sum_{d_0 = k}^d \Pr ( \Bin(d,p) = \vec d_0) \binom{\vec d_0}{k} x_{\distactive{v}}^{\vec d_0 - k}(1-x_{\distactive{v}})^k & \text{if }v = \lca\\ 
        \sum_{d_0 = 1}^d \Pr ( \Bin(d,p) = \vec d_0) \binom{\vec d_0}{1}x_{\distactive{v}}^{\vec d_0 -1}(1-\distactive{v}) & \text{if }v \not = \lca, v \in \cC. \\
    \end{cases}
\end{align*}
}
\end{proposition}
\begin{proof}
The first part, namely that $x_t = (1-p+p x_{t-1})^{d-1}$, is an application of probability generating functions and their connection to branching processes. The probability generating function $f_{\Bin(n,p)}$ of a Binomial distribution with parameters $n$ and $p$ reads
$$ f_{\Bin(n,p)}(s) = \Erw \brk{ s^{ \Bin(n,p) } } = \bc{ 1 - p + ps }^n. $$
Now, let $p_k = \Pr \bc{ \Bin(d-1, p) = k}$. Then, it is immediate that
\begin{align*}
    x_{t+1} = p_0 + p_1 x_{t-1} + \ldots + p_{d-1} x_{t-1}^{d-1} = \sum_{k = 0}^{d-1} p_k x_{t-1}^k 
\end{align*}
which is exactly the probability generating function of the Binomial distribution. A detailed explanation and a formal proof of this statement can be, for instance, found in  \cite{branching_2020}. 

Let us now suppose that $v = \lca$. Let $\cV_0$ be the event that $v$ has exactly $k \leq \vec d_0 \leq d$ children that get activated by $v$. 
Clearly, $\Pr \bc{ \cV_0 } = \Pr \bc{ \Bin \bc{ d, p } = \vec d_0 }$ and of course, $\vec d_0$ needs to be at least $k$ as differently, the probability of having $k$ active sub-trees was zero. 
Given $\cV_0$, we find that the probability of observing exactly $k$ active sub-trees is the probability that exactly $k$ out of $\vec d_0$ independent Galton-Watson processes with offspring distribution $\Bin(d-1, p)$ survived the first $\distactive{v}$ steps. 
Therefore, the number of active sub-trees at time $t$ is distributed as $ \Bin \bc{ \vec d_0, x_{\distactive{v}} } $ given $\cV_0$ and the first part of the formula follows.

If, on contrary, $v$ is not the closest candidate but a node that has a different distance from $\activenodes$, we observe that from the originally $1 \leq \vec d_0 \leq d$ Galton-Watson processes originated in the children of $v$, exactly one process needed to survive and $\vec d_0 - 1$ needed to be extinct at time $\distactive{v}$. The proposition follows.
\end{proof}

\paragraph{Proof of \cref{thm_lca_dreg} (i).}
To prove the first part of \cref{thm_lca_dreg}, it suffices to apply the first part of \cref{prop_dreg}. Indeed, if $(d-1)p \leq 1$, we find that the smallest fixed-point of $x \mapsto (1 - p + p x)^{d-1}$ is $x = 1$. Therefore, $x_t = 1 - o_t(1)$. Furthermore, as $p$ is a constant, we have that in this case $d = \Theta(1)$ as well. Therefore, a union bound over the at most $d$ possible independent Galton-Watson processes with offspring distribution $\Bin(d-1, p)$ originated in the children of $\pz$, yields that with probability $1 - o_t(1)$, we find $\activenodes = \emptyset$. In this case, detection clearly fails with high probability.
\qed

\paragraph{Proof of \cref{thm_lca_dreg} (ii).}
We start with the part of the theorem that claims that weak detection succeeds by the source detection heuristic, namely that $\Pr \bc{\lca = \pz} = \Omega(1)$. We find that $\lca = \pz$ with probability one if the set of possible candidate nodes $\cC$ has size 1. Therefore, it suffices to prove that $\Pr \bc{ \abs{ \cC } = 1 } = \Omega(1)$. This is the case if (and only if), the rumor's source $\pz$ propagated the rumor to all of its $d$ children and all $d$ independent Galton-Watson processes with offspring distribution $\Bin(d-1, p)$ originated in the children of $\pz$ did not become extinct. Let $\vec d_0$ denote the number of children of $\pz$ that get activated. Clearly,
\begin{align} \label{eq_dreg_ii_i}
    \Pr \bc{ \vec d_0 = d } = \Pr \bc{ \Bin(d, p) = d } = p^d = \Omega(1)
\end{align}
as, by assumption, $p$ and $d$ are constants.
Furthermore, since $1 < (d-1)p = \Theta(1)$ holds, the smallest fixed point of $x \mapsto (1 - p + p x)^{d-1}$ is a real number between zero and one. Therefore, by \cref{prop_dreg}, there are constants $0 < \gamma_1 \leq \gamma_2 < 1$ such that
\begin{align} \label{eq_dreg_ii_ii}
    \gamma_1 - o_t(1) < x_t < \gamma_2 + o_t(1). 
\end{align}
Thus, it follows that the source detection heuristic succeeds at weak detection if $1 < (d-1)p = \Theta(1)$ with probability $1 - o_t(1)$ by \eqref{eq_dreg_ii_i} -- \eqref{eq_dreg_ii_ii}.

It is left to prove that under the same assumptions we have
\begin{align} \label{eq_dreg_ii_iii}
\Pr \bc{ \dist \bc{ \lca, \pz } \geq k  } \leq \exp \bc{ - \Omega(k) }.
\end{align}
Suppose that $ \abs{\distactive{\pz} - \distactive{\lca} } = \dist \bc{ \lca, \pz } = k > 3$. Therefore, there is a unique path $P_{\lca, \pz}$ in $G$ that connects $\lca$ and $\pz$ with $k-2$ internal vertices. All of those internal vertices will get activated exactly ones during the spreading process. Therefore, for any of those $k-2$ steps, the process needs to either activate only exactly one child or, it activates more than one child, but the remaining processes died out at the observation time. Of course, we have that the probability that a node spawns exactly one active child is given by
\begin{align} \label{eq_dreg_onechild}
    \Pr \bc{ \Bin( d, p ) = 1 } = p (1-p)^{d-1} 
\end{align}
which is a real number bounded away from zero and one if $d, p = \Theta(1)$.
By \eqref{eq_dreg_ii_i} -- \eqref{eq_dreg_ii_ii} as well as \eqref{eq_dreg_onechild}, we find that there is a sequence of constants $\cbc{\gamma_i}_{i=1 \ldots k-2}$ dependent only on $p, d,$ and $k$ all of which $\gamma_i$ are uniformly bounded away from zero and one. Therefore,  
$$\Pr \bc{ \dist \bc{ \lca, \pz } \geq k  } \leq \min_{i=1 \ldots k-2} \gamma_i^k = \exp \bc{ - \Omega(k) } $$
which implies \eqref{eq_dreg_ii_iii}.
\qed

\paragraph{Proof of \cref{thm_lca_dreg} (iii).}
The last part of \cref{thm_lca_dreg} (iii) states that the source detection heuristic succeeds at strong detection with high probability if $(d-1)p \in \omega(1)$. We start with the following simple observation which is an immediate consequence of the Chernoff bound applied to the $\Bin \bc{(d-1), p}$ distribution.
\begin{fact} \label{lem_dreg_iii_i}
If $(d-1)p \in \omega(1)$, we find that if node $v$ gets activated, the number of activated children $\omega_v$ satisfies $\omega_v \geq \frac{(d-1)p}{2} = \omega(1)$ with high probability. 
\end{fact}
As a next observation, we claim that if $\vec d_0 = \omega(1)$ Galton-Watson processes with offspring distribution $\Bin \bc{(d-1), p}$ start independently, at least $(1 - \eps)$-fraction of them will not become extinct eventually with high probability for any $\eps > 0$.
\begin{lemma} \label{lem_dreg_iii_ii}
Suppose that $\ell$ Galton-Watson processes with offspring distribution $\Bin \bc{(d-1), p}$ start independently under the condition that $(d-1)p \in \omega(1)$. Let $\vec Y$ denote the number of processes that do not ultimately become extinct. Then, $\Pr \bc{ \vec Y \geq (1 - o(1)) \ell} \geq 1 - o(1)$.
\end{lemma}
\begin{proof}
By \cref{prop_dreg}, we have that the probability that one of the processes becomes extinct is $p_e = o(1)$. Thus, the number of not-extinct processes is Binomially distributed with parameters $\ell$ and $1 - p_e$. Therefore, the lemma follows from the Chernoff bound.
\end{proof}

By \cref{lem_dreg_iii_i,lem_dreg_iii_ii} we directly find that, with high probability, all but $o(dp)$ of the processes started in the children of $\pz$ are still active at the observation time. 

Suppose that $\pz \neq \lca$ and $\dist(\pz, \lca) = k \geq 1$. This implies that either $k$ times only exactly one child is activated or, given that multiple children are activated, only exactly one of those spreading processes survived eventually. For a specific step $1 \leq i \leq k-1$, the probability that this occurs is by \cref{lem_dreg_iii_i,lem_dreg_iii_ii} at most $\gamma_i = o(1)$. Therefore,
$\Pr \bc{ \lca \neq \pz \mid \activenodes} = o(1)$
which implies the third part of \cref{thm_lca_dreg}.
\qed 

\subsection{Proof of \cref{thm_lca_gw}}
The main proof strategy is analogous to the proof of \cref{thm_lca_dreg} with one fundamental difference which makes the analysis more involved. While in the $d$-regular case, a once activated node spawns a random number of independent Galton-Watson processes with offspring distribution $\Bin(d-1, p)$, this is not true in the setting of \cref{thm_lca_gw}. Here, the underlying network itself is a Galton-Watson process with offspring distribution $\Po(\lambda)$. Fortunately, the \emph{Poisson thinning principle} \cite{Kingman1992} implies that, in distribution, we can analyze the following spreading process. Once $v$ gets activated, it spawns $\vec d_0 \sim \Po(\lambda p)$ active children and thus $\vec d_0$ independent Galton-Watson processes with offspring distribution $\Po(\lambda p)$. The extinction probability $\bar{x_t}$ can be derived similarly as in \cref{prop_dreg} and reads $\bar{x}_t = \exp \bc{- \lambda p (1 - \bar{x}_{t-1})}$.

Now, we have all ingredients to prove (i) -- (iii) of \cref{thm_lca_gw}. All proofs are based on similar ideas as their counter-parts in the proof of \cref{thm_lca_dreg}. The main challenge arising is that, in comparison to the Binomial distribution, the Poisson distribution has heavy tails which must be taken into account. This challenge is technically non-trivial but standard concentration tools for large deviation analyses suffice in order to prove the result. All omitted proof details can be found in the appendix.
\qed

\section{Conclusion} 
We pin down exact information-theoretic phase-transitions in the source detection task on important tree-network models by proving that as soon as weak detection is possible information-theoretically, the efficient closest candidate heuristics succeeds at this task. Those findings imply, of course, the same result for $\mathbb{G}(n,p)$ and random $d$-regular graphs as long as they are sufficiently sparse and the spreading process ran for only $o(\log(n))$ steps as those random networks are then, locally, given by the described tree-networks with high probability.

Furthermore, we show empirically that the source detection heuristic performs well on non-acyclic networks and seems to be a very decent and efficient estimator of the rumor's source. A natural question is whether it is possible to prove similar information-theoretic bounds for non-acyclic networks. While this seems to be a very challenging task in general, it might become accessible if we restrict ourselves to specific random networks or networks with a specific tree-width. 

Finally, on the empirical side, it is an interesting question whether the rumor's source of the Independent Cascade Process can be learned by graph neural networks. This seems challenging as only few vertices are active and the network would need to learn possible propagation paths in a graph given only a snapshot of the network.

\begin{ack}
PB, MHK, and MR are funded by the German Research Council (DFG FOR 2975).
\end{ack}

\bibliographystyle{abbrv}
\bibliography{bibliography}

\appendix

\section{Appendix}

\subsection{Proof in the Galton-Watson model, \cref{thm_lca_gw}}
The proof is similar to the proof of \cref{thm_lca_dreg}. The most fundamental difference is that in the $d$-regular case, a once activated node spawns a random number of independent Galton-Watson processes with offspring distribution $\Bin(d-1, p)$, this is not true if the underlying network itself is a Galton-Watson process with offspring distribution $\Po(\lambda)$. The \emph{Poisson thinning principle} \cite{Kingman1992} allows to analyze the resulting distribution.

\begin{fact}[Application of the Poisson thinning principle]
Let $\vec X \sim \Po(d)$ and furthermore, given $\vec X$, define $\vec Y = \Bin \bc{ \vec X, p }$. Then $\vec Y \sim \Po( \lambda p )$.
\end{fact}
Thus, once $v$ gets activated, it spawns $\vec d_0 \sim \Po(\lambda p)$ active children and thus $\vec d_0$ independent Galton-Watson processes with offspring distribution $\Po(\lambda p)$. The following proposition characterizes the extinction probability of such processes.

\begin{proposition}\label{prop_gwa}
If $\bar{x}_t$ is the probability that at time $t$ there are no more active vertices under the spreading model on a super-critical Galton-Watson tree with offspring distribution $\Po (\lambda p)$, we find
\[\bar{x}_t = \exp \bc{- \lambda p (1 - \bar{x}_{t-1})} \]
and the ultimate extinction probability of the spreading process is the smallest fixed-point of $\bar{x} \mapsto \exp \bc{- \lambda p (1 - \bar{x})}$.
Furthermore
\begin{align*}
    \Pr \bc{ Y_v \bc{ \activenodes } = k \mid \pz = v} = 
    \begin{cases}
        \sum_{d_0 = k}^d \Pr ( \Po( \lambda p ) = \vec d_0) \binom{\vec d_0}{k} x_{\distactive{v}}^{\vec d_0 - k}(1-x_{\distactive{v}})^k & \text{if }v = \lca\\ 
        \sum_{d_0 = 1}^d \Pr ( \Po( \lambda p ) \binom{\vec d_0}{1}x_{\distactive{v}}^{\vec d_0 -1}(1-\distactive{v}) & \text{if }v \not = \lca, v \in \cC \\
    \end{cases}
\end{align*}
\end{proposition}

\begin{proof}
As in the proof of \cref{prop_dreg}, the recurrence $\bar{x}_t = \exp \bc{- \lambda p (1 - \bar{x}_{t-1})}$ can be easily calculated by the probability generating function of the Poisson distribution. Indeed, let $f_{\Po(\lambda)}(s) = \Erw \brk{ s^{\Po(\lambda)} }$ be the probability generating function of the Poisson distribution. It is well known that
\begin{align*}
    f_{\Po(\lambda)}(s) = \exp \bc{ - \lambda (1 - s) }.
\end{align*}
Again, we refer to \cite{branching_2020} for a detailed explanation of the connection between the probability generating function and the extinction probability of branching processes.

Now, for brevity, suppose that $v = \lca$. Let $\cV_0$ be the event that $v$ has exactly $k \leq \vec d_0 \leq d$ children that get activated by $v$. Similarly as before, $\Pr \bc{ \cV_0 } = \Pr \bc{ \Po( \lambda p ) = \vec d_0 }$ and of course, $\vec d_0$ needs to be at least $k$ as differently, the probability of having $k$ active sub-trees was zero. 

Given $\cV_0$, we again start $\vec d_0$ independent Galton-Watson processes with offspring distribution $\Po( \lambda p )$ in the children. Therefore, the probability of observing exactly $k$ active sub-trees is the probability that exactly $k$ out of $\vec d_0$ of those processes are not extinct after $\distactive{v}$ steps. Of course, the number of such active sub-trees at time $t$ is distributed as $ \Bin \bc{ \vec d_0, \bar{x}_t } $ given $\cV_0$ and the first part of the formula follows.

As in the $d$-regular case, if on contrary $v$ is not the closest candidate but a node further apart from $\activenodes$, we observe that from the originally $1 \leq \vec d_0 \leq d$ Galton-Watson processes originated in the children of $v$, exactly one process needed to survive and $\vec d_0 - 1$ needed to be extinct at time $\distactive{v}$.
\end{proof}

\paragraph{Proof of \cref{thm_lca_gw} (i).}
As in the $d$-regular case, the first part of \cref{thm_lca_gw} follows by the first part of \cref{prop_gwa}. If $\lambda p \leq 1$, the smallest fixed-point of $\bar{x} \mapsto \exp{ - \lambda p (1 - \bar{x}) }$ is $\bar{x} = 1$. Therefore, $\bar{x}_t = 1 - o_t(1)$ describes the probability that the underlying spreading process died out until time-step $t$. More precisely, by the recurrence equation, we find the following. Suppose that $\eps_t = o_t(1)$ denotes the convergence speed towards $1$. Then, by the recurrence equation and a Taylor approximation we have
\begin{align*}
    1 - \eps_t &= 1 - \lambda p \bc{\eps_{t-1} + \frac{\lambda^2 p^2 \eps_{t-1}^2}{2} } + O \bc{ \eps_{t-1}^3 }.
\end{align*}
If $\lambda p < 1$, we directly find that $\eps_t = O \bc{ (\lambda p)^t }$ decays exponentially fast in $t$. If $\lambda p = 1$, this is much more subtle. Indeed, we find
\begin{align*}
    \eps_t & \leq \bc{\eps_{t-1} - \frac{\eps_{t-1}^2}{2} } + O \bc{ \eps_{t-1}^3 }
\end{align*}
and therefore, we only get $\eps_t = O \bc{ t^{-1} }$ in this case.

Since we assume $p$ to be a constant, clearly $\lambda = O(1)$ as well. Unfortunately, the Poisson tails are kind of heavy. More precisely, even if $\lambda$ is a constant, the probability that a $\Po(\lambda)$ variable becomes large is not negligible. We analyze this by a careful application of limits. Recall that we assume that the underlying tree-network is infinite. We model this as follows. Suppose that the tree-network consists of $n$ vertices and we will let $n \to \infty$.

Let $C > 0$, then the probability that the number of neighbors of a specific node $v$ exceeds $C$ is, for large $C$, given by Chernoff's bound as
\begin{align*}
    \Pr \bc{ \abs{ \partial v } > C } = \Pr \bc{ \Po(\lambda) > C } \leq \exp \bc{ - \frac{ (C - \lambda)^2 }{ 2 C } } \sim \exp \bc{ - C/2 }. 
\end{align*}
As the number of spawned children is independent for all vertices, the number of vertices of degree at least $C$ is stochastically dominated by a $\Bin \bc{n, \exp \bc{-C/2}}$. Thus, with probability $1 - o_n(1)$, there are no more than $O( n \sqrt{\log(n)} \exp \bc{ - D } )$ vertices of degree $D > 0$ for a sufficiently large constant $D$ (independent of $n$) if $n \to \infty$. 

We denote by $\cD$ the event that this is actually true. Thus, conditioned on $\cD$, there are only $O( n \sqrt{\log(n)} \exp \bc{ - D } )$ vertices of degree at least $D$. Now, we chose $\pz$ uniformly at random out of all vertices. Therefore, given $\cD$, the probability that $\pz$ has small degree is given by,
\begin{align*}
    \Pr \bc{ \abs{ \partial \pz } > D \mid \cD} = 1 - O \bc{ \frac{ \sqrt{\log(n)} }{\exp \bc{ - D }} }.
\end{align*}
Clearly, this becomes only a high probability event if $D = \Omega \bc{ \log \log n}$. In the worst case, we find that a union bound over all activated children of $\pz$ leads only to ultimate extinction of all processes, if 
\begin{align*}
    O \bc{ \frac{\log \log n}{t} } = o_t(1),
\end{align*}
or, differently, that $t = \Omega(\log \log n)$. As in the theorem, we only claim the assertion in the limit $t \to \infty$ and we assume the underlying tree-network to be infinite, this proves the claim of the theorem. We remark at this point that the assumption that $t$ depends slightly on $n$ does no harm to applications as, on real networks, $\log \log n$ can be seen as a constant.
\qed

\paragraph{Proof of \cref{thm_lca_gw} (ii).}
Again, as in the $d$-regular case, we start proving the weak detection property of the source detection heuristic. Thus, we aim to prove $\Pr \bc{ \abs{ \cC } = 1 } = \Omega(1)$. 

This is the case if (and only if) $\pz$ propagated the rumor to all of its $\vec d_{\pz} \Po \bc{ \lambda }$ children and all $\vec d_0$ independent Galton-Watson processes with offspring distribution $\Po \bc{\lambda p}$ rooted at the children of $\pz$ did die out eventually. Let $\vec d_0$ denote the number of children of $\pz$ that get activated. We first need to calculate the probability that $\vec d_0 = \vec d_{\pz}$. To this end, let
\begin{align*}
    I_0(x) & = \sum_{k=0}^\infty \frac{1}{k! \Gamma(k + 1)} \bc{\frac{x}{2}}^{2k}
\end{align*}
denote the modified Bessel function of order zero. It is well known (see, for instance, Equation (2) of \cite{besselfunctions}) that
\begin{align*}
    I_0(x) &= \frac{\exp \bc{x}}{\sqrt{2 \pi x}} \bc{ 1 + \frac{1}{8x} + \frac{9}{128 x^2} + O \bc{\frac{1}{x^3}} }.
\end{align*}

We have
\begin{align*} 
    \Pr \bc{ \vec d_0 = \vec d_{\pz} \mid \vec d_{\pz}} = \frac{ (\lambda p)^{\vec d_{\pz}} \exp \bc{- \lambda p}}{\vec d_{\pz}!}. 
\end{align*}
Therefore, by the law of total probability,
\begin{align} \label{eq_gw_ii_i}
    \Pr & \bc{ \text{all children of $\pz$ get activated} }  = \sum_{k = 1}^\infty \Pr \bc{\vec d_{\pz} = k } \frac{ (\lambda p)^{k} \exp \bc{- \lambda p}}{k!} \notag \\ 
    & = \sum_{k = 1}^\infty \frac{\lambda^k \exp \bc{ - \lambda}}{k!} \frac{ (\lambda p)^{k} \exp \bc{- \lambda p}}{k!} = \exp \bc{ - \lambda (1 + p) } \sum_{k = 1}^\infty \frac{ \bc{\lambda^{2} p}^k }{ (k!)^2} \notag \\
    & = \exp \bc{ - \lambda (1 + p) } \bc{ I_0( 2 \lambda \sqrt{p} ) - 1 } = \exp \bc{ - \lambda (1 + p) } \bc{ \frac{\exp \bc{ 2 \lambda \sqrt{p} }}{ \sqrt{4 \pi \lambda \sqrt{p}} } - 1 + O \bc{ \bc{\lambda \sqrt{p}}^{-1} } } \notag \\
    & = \frac{\exp \bc{ - \lambda (1 - \sqrt{p})^2} }{\sqrt{4 \pi \lambda \sqrt{p}}} - \exp \bc{ - \lambda (1 + p) } + O \bc{ \exp \bc{ - \lambda (1 + p) } \bc{\lambda \sqrt{p}}^{-3/2} }.
\end{align}
If $\lambda$ and $p$ are constants, it is immediate from \eqref{eq_gw_ii_i} that there is a constant $\gamma > 0$ such that
\begin{align}
    \Pr & \bc{ \text{all children of $\pz$ get activated} } > \gamma.
\end{align}
Finally, since $1 < \lambda p = \Theta(1)$ by assumption, the smallest fixed point of $\bar{x} \mapsto \exp \bc{ - \lambda(1 - \bar{x}) }$ is a real number between zero and one. Therefore, by \cref{prop_gwa}, there are constants $0 < \gamma_1 \leq \gamma_2 < 1$ such that
\begin{align} \label{eq_gw_ii_ii}
    \gamma_1 - o_t(1) < \bar{x}_t < \gamma_2 + o_t(1). 
\end{align}
Therefore, the source detection heuristic succeeds at weak detection if $1 < \lambda p = \Theta(1)$ with probability $1 - o_t(1)$ by \eqref{eq_gw_ii_i} -- \eqref{eq_gw_ii_ii}.

Again, we are left to prove the decay property, namely
\begin{align} \label{eq_gw_ii_iii}
\Pr \bc{ \dist \bc{ \lca, \pz } \geq k  } \leq \exp \bc{ - \Omega(k) }.
\end{align}
Suppose that $ \abs{\distactive{\pz} - \distactive{\lca} } = \dist \bc{ \lca, \pz } = k > 3$. As previously, we find a unique path $P_{\lca, \pz}$ in $G$ that connects $\lca$ and $\pz$ with $k-2$ internal vertices. Exactly as in the $d$-regular case, of those internal vertices will get activated exactly ones during the spreading process and so, in every step, the process needs to either activate exactly one child or all but one of the remaining processes ultimately are extinct. The probability that a node spawns exactly one child is given by
\begin{align} \label{eq_gw_onechild}
    \Pr \bc{ \Po( \lambda p ) = 1 } = \lambda \exp(- \lambda) 
\end{align}
which is a real number bounded away from zero and one if $\lambda = \Theta(1)$.
By \eqref{eq_gw_ii_i} -- \eqref{eq_gw_ii_ii} and \eqref{eq_gw_onechild}, there is a sequence of constants $\cbc{\gamma_i}_{i=1 \ldots k-2}$ dependent only on $\gamma(p, \lambda, k)$ and uniformly bounded away from zero and one, such that  
$$\Pr \bc{ \dist \bc{ \lca, \pz } \geq k  } \leq \min_{i=1 \ldots k-2} \gamma_i^k = \exp \bc{ - \Omega(k) } $$
and \eqref{eq_gw_ii_iii} follows.
\qed

\paragraph{Proof of \cref{thm_lca_gw} (iii).}
In order to prove the last part of \cref{thm_lca_gw}, we start with the following observation.
\begin{lemma} \label{lem_gw_iii_i}
If $\lambda p \to \infty$, then, with high probability, an activated node $v$ satisfies the following.
\begin{itemize}[nosep]
    \item $\deg(v) \geq \frac{\lambda}{2} = \omega(1)$ with high probability,
    \item The number of activated children $\vec \omega_v$ satisfies $\omega_v \geq \frac{\lambda}{2} = \omega(1)$ with high probability. 
\end{itemize}
\end{lemma}
\begin{proof}
This is an immediate consequence of the Chernoff bound applied to the $\Po(\lambda)$ and, respectively, $\Po(\lambda p)$ distribution given that $\lambda p \to \infty$.
\end{proof}
As a next observation, we claim that if $\vec \omega_v$ (as given by \cref{lem_gw_iii_i}) Galton-Watson processes with offspring distribution $\Po(\lambda p)$ are initialized independently, at least $(1 - \eps)$-fraction will survive with high probability for any $\eps > 0$.
\begin{lemma} \label{lem_gw_iii_ii}
Suppose that $X$ Galton-Watson processes with offspring distribution $\Po(\gamma)$ start independently under the condition that $\gamma \to \infty$. Let $\vec Y$ denote the number of such processes that did not ultimately become extinct. Then, $\Pr \bc{ \vec Y \geq (1 - o(1)) X} \geq 1 - o(1)$.
\end{lemma}
\begin{proof}
By \cref{prop_gwa}, the probability that one specific process out of the $X$ processes gets extinct is $p_e = o(1)$. Thus, as in the $d$-regular case, we have that the number of not-extinct processes is $\Bin(X, 1-p_e)$-distributed. The lemma follows from Chernoff's bound.
\end{proof}

As in the $d$-regular case, the previous lemmas imply that with high probability $\omega(1)$ of the processes started in the children of $\pz$ are still active. 

Suppose that $\pz \neq \lca$ and $\dist(\pz, \lca) = k \geq 1$. This implies that either $k$ times only one active child is spawned or, if multiple children are spawned, only exactly one rumor spreading process rooted in those children survives eventually. But by \cref{lem_gw_iii_i,lem_gw_iii_ii}, we find that this probability is $o(1)$ for all $k \geq 1$.
\qed 

\clearpage
\section{Additional Simulation Data}

\label{apx:additional-data}
In this appendix we present the additional simulation data for random regular graphs (configuration model) with $d = 4$ (see \cref{tab:table2}) and random geometric graphs with an expected node degree of $16$ (see \cref{tab:table3}). We initialized both networks with $n = 10^5$ nodes.

\begin{table}[H]
\centering
\caption{Simulation results for random regular graphs (configuration model).}
\label{tab:table2}
 \small
\begin{tabular}{lrrrrr}
\toprule
$p$ & \parbox{\widthof{number of}}{\raggedleft number of\newline successes} & $\cc \neq \pz$ & $\activenodes = \emptyset$ & \parbox{\widthof{distance}}{\raggedleft average\newline distance} & \parbox{\widthof{maximum}}{\raggedleft maximum\newline distance}\\
\midrule
0.00 & 0 & 0 & 100 & - & - \\
0.05 & 0 & 0 & 100 & - & - \\
0.10 & 0 & 0 & 100 & - & - \\
0.15 & 0 & 0 & 100 & - & - \\
0.20 & 0 & 0 & 100 & - & - \\
0.25 & 0 & 5 & 95 & 7.20 & 9 \\
0.30 & 0 & 16 & 84 & 7.53 & 11 \\
0.35 & 2 & 26 & 72 & 6.86 & 12 \\
0.40 & 19 & 43 & 38 & 5.37 & 11 \\
0.45 & 38 & 40 & 22 & 2.70 & 11 \\
0.50 & 43 & 41 & 16 & 2.06 & 9 \\
0.55 & 70 & 23 & 7 & 0.57 & 6 \\
0.60 & 76 & 21 & 3 & 0.34 & 4 \\
0.65 & 86 & 11 & 3 & 0.15 & 3 \\
0.70 & 87 & 10 & 3 & 0.14 & 3 \\
0.75 & 98 & 2 & 0 & 0.02 & 1 \\
0.80 & 97 & 3 & 0 & 0.03 & 1 \\
0.85 & 99 & 1 & 0 & 0.01 & 1 \\
0.90 & 100 & 0 & 0 & 0 & 0 \\
0.95 & 100 & 0 & 0 & 0 & 0 \\
1.00 & 100 & 0 & 0 & 0 & 0 \\
\bottomrule
\end{tabular}
\end{table}

\begin{table}[H]
\centering
\caption{Simulation results for random geometric graphs (expected node degree $16$).}
\label{tab:table3}
 \small
\begin{tabular}{lrrrrr}
\toprule
$p$ & \parbox{\widthof{number of}}{\raggedleft number of\newline successes} & $\cc \neq \pz$ & $\activenodes = \emptyset$ & \parbox{\widthof{distance}}{\raggedleft average\newline distance} & \parbox{\widthof{maximum}}{\raggedleft maximum\newline distance}\\
\midrule
0.00 & 0 & 0 & 100 &  & 0 \\
0.05 & 1 & 1 & 98 & 2.25 & 4 \\
0.10 & 5 & 29 & 66 & 2.39 & 5 \\
0.15 & 25 & 59 & 16 & 2.02 & 6 \\
0.20 & 34 & 63 & 3 & 1.68 & 5 \\
0.25 & 51 & 48 & 1 & 1.51 & 4 \\
0.30 & 62 & 36 & 2 & 1.40 & 5 \\
0.35 & 71 & 29 & 0 & 1.24 & 5 \\
0.40 & 86 & 14 & 0 & 1.11 & 4 \\
0.45 & 94 & 6 & 0 & 1.04 & 3 \\
0.50 & 94 & 6 & 0 & 1.13 & 5 \\
0.55 & 95 & 5 & 0 & 1.07 & 5 \\
0.60 & 100 & 0 & 0 & 0.97 & 4 \\
0.65 & 95 & 5 & 0 & 1.03 & 6 \\
0.70 & 99 & 1 & 0 & 0.79 & 3 \\
0.75 & 99 & 1 & 0 & 1.03 & 5 \\
0.80 & 100 & 0 & 0 & 0.97 & 5 \\
0.85 & 100 & 0 & 0 & 0.96 & 6 \\
0.90 & 100 & 0 & 0 & 0.66 & 3 \\
0.95 & 98 & 2 & 0 & 0.87 & 5 \\
1.00 & 100 & 0 & 0 & 0.85 & 6 \\
\bottomrule
\end{tabular}
\end{table}

\end{document}